\documentclass[11pt]{article}

\usepackage[utf8]{inputenc}

\usepackage{amsmath,amsfonts,amsthm,amssymb,color}
\usepackage[usenames,dvipsnames,svgnames,table]{xcolor}
\definecolor{darkgreen}{rgb}{0.0,0,0.9}
\usepackage{mathtools}
\usepackage{authblk}
\usepackage{fullpage}
\usepackage{parskip}
\usepackage{comment}
\usepackage{tikz}
\usepackage{bbm}
\usepackage{dsfont}
\usepackage[sc]{mathpazo}
\usepackage[basic]{complexity}
\usepackage{algorithm, algpseudocode}
\usepackage[colorlinks=true,
citecolor=OliveGreen,linkcolor=BrickRed,urlcolor=BrickRed,pdfstartview=FitH]{hyperref}
\usepackage{cleveref}

\newcommand*{\suppress}[1]{}

\newcommand*{\abs}[1]{\lvert#1\rvert}

\makeatletter
\def\thm@space@setup{%
	\thm@preskip= 10pt
	\thm@postskip=\thm@preskip 
}
\makeatother

\makeatletter
\renewcommand{\paragraph}{%
	\@startsection{paragraph}{4}%
	{\z@}{5pt}{-1em}%
	{\normalfont\normalsize\bfseries}%
}
\makeatother

\newtheorem{theorem}{Theorem}
\newtheorem{lemma}{Lemma}

\newtheorem{proposition}{Proposition}

\theoremstyle{definition}

\newenvironment{fminipage}%
{\begin{Sbox}\begin{minipage}}%
		{\end{minipage}\end{Sbox}\fbox{\TheSbox}}

\def\abs#1{\left|#1  \right|}

\newcommand\ff{\boldsymbol{\mathit{f}}}

\newcommand\ww{\boldsymbol{\mathit{w}}}
\newcommand\yy{\boldsymbol{\mathit{y}}}

\newcommand\xx{\boldsymbol{\mathit{x}}}

\newcommand{\Mc}{\mathcal{M}}
\newcommand{\Mcc}{\mathcal{M'}}
\newcommand{\Lc}{\mathcal{L}}
\newcommand{\Lcc}{\mathcal{L'}}
\newcommand{\ST}{\text{s.t.}}
\newcommand{\ffbar}{\overline{\ff}}
\newcommand{\Sbar}{\overline{S}}
\newcommand{\Sbarc}{\overline{S'}}
\newcommand{\fbar}{\overline{f}}
\newcommand{\Pibar}{\overline{\Pi}}

\title{A Natural Generalization of Stable Matching\\ 
Solved via New Insights into Ideal Cuts}

\author[1]{Tung Mai}
\author[2]{Vijay V.~Vazirani}

\affil[1]{Georgia Tech}
\affil[2]{University of California, Irvine}

\date{}

\begin{document}
	\maketitle
    
	\begin{abstract}
We study a natural generalization of stable matching to the maximum weight stable matching problem
and we obtain a combinatorial polynomial time algorithm for it by reducing it to the 
problem of finding a maximum weight ideal cut in a DAG. We give the first polynomial time 
algorithm for the latter problem; this algorithm is also combinatorial.

The combinatorial nature of our
algorithms not only means that they are efficient but also that 
they enable us to obtain additional structural and algorithmic results:
\begin{itemize}
	\item We show that the set, $\Mcc$, of maximum weight stable matchings forms a sublattice
$\Lcc$ of the lattice $\Lc$ of all stable matchings $\Mc$.
\item
We give an efficient algorithm for finding boy-optimal and girl-optimal matchings in $\Lcc$.
\item
We generalize the notion of rotation, a central structural notion in the context of the 
stable matching problem, to {\em meta-rotation}. Just as rotations help traverse the lattice
$\Lc$, meta-rotations help traverse the sublattice $\Lcc$.
\end{itemize}

\end{abstract}

\section{Introduction}

The two problems of stable matching and cuts in graphs were introduced in the seminal papers
of Gale and Shapley (1962) \cite{GaleS} and Ford and Fulkerson (1956) \cite{FF},
respectively. Over the decades, remarkably deep and elegant theories have emerged around both
these problems which include highly sophisticated efficient 
algorithms, not only for the basic problems but also several generalizations and variants,
that have found numerous applications \cite{Knuth-book,GusfieldI,Manlove-book,Sch-book}.

In this paper, we study a natural 
generalization of stable matching to the maximum weight stable matching problem
and we obtain an efficient combinatorial algorithm for it; we
remark that the linear programming formulation of stable matching can be used to show that the
weighted version is in P, see Section \ref{sec.related}. 
Our algorithm is obtained by reducing this problem to the problem
of finding a maximum weight ideal cut in a DAG. We give the first polynomial time algorithm
for the latter problem; this algorithm is also combinatorial. The combinatorial nature of our
algorithms not only means that they are efficient but also that 
they enable us to obtain additional structural and algorithmic results:
\begin{itemize}
	\item We show that the set, $\Mcc$, of maximum weight stable matchings forms a 
sublattice $\Lcc$ of the stable matching lattice $\Lc$.
\item
We give an efficient algorithm for finding boy-optimal and girl-optimal matchings in $\Mcc$.
\item
We generalize the notion of rotation, a central structural notion in the context of the 
stable matching problem, to {\em meta-rotation}. Analogous to the way rotations help traverse 
the lattice $\Lc$, meta-rotations help traverse the sublattice $\Lcc$.
\end{itemize}

The maximum weight stable matching problem has several applications, e.g., the 
{\em egalitarian stable matching problem} defined 
in section \ref{sec.related}. Another application is: given a set $D$ of desirable boy-girl 
pairs and a set $U$ of undesirable pairs, find a stable matching that simultaneously 
maximizes the number of pairs in $D$ and minimizes the number of pairs in $U$. This reduces to
our problem by assigning each pair in $D$ a weight of 1 and each pair in $U$ a weight of $-1$.

Our results are based on deep properties of rotations and the manner in which closed sets 
in the rotation poset $\Pi$ yield stable matchings in the lattice $\Lc$.

\subsection{Problem definitions}

Let $I$ denote an instance of the stable matching problem over sets $B$ and $G$ of
$n$ boys and $n$ girls, respectively. Let $\ww$ be a
weight function $\ww: \ B \times G \rightarrow \mathbb{Q}$. Then $(I, w)$ defines an instance
of the {\em maximum weight stable matching problem}; it asks for a stable matching of instance
$I$, say $M$, that maximizes the objective function 
$ \sum_{bg \in M} w_{bg} .$

\suppress{
In the {\em bi-objective stable matching problem} we are given sets $B$ and $G$, of $n$ boys and
$n$ girls and, for each boy and each girl,
a complete preference ordering over all agents of the opposite sex. \
However, unlike the previous problem,
we are given two weight functions
$\ww^{(1)},\ww^{(2)}: \ B \times G \rightarrow \mathbb{R}$. The problem is to find a stable matching $M$ that 
maximizes 
$ \sum_{bg \in M} w^{(2)}_{bg} $
among the ones maximizing  
$\sum_{bg \in M} w^{(1)}_{bg}.$
}

In the {\em maximum weight ideal cut problem} we are given a directed acyclic graph $G = (V,E)$ 
with a source $s$ and a sink $t$ such that for each $v \in V$, there is a path from $s$ to $v$ and 
a path from $v$ to $t$. We are also given a weight $w_{uv} \in \mathbb{Q}$ for each edge
$(u, v) \in E$. An {\em ideal cut} is a partition of the vertices into sets 
$S$ and $\overline{S} = V(G) \setminus S$ such that $s \in S$ and $t \in \overline{S}$ and
there is no edge $uv \in E$ with $u \in \overline{S}$ and $v \in S$. 
We remark that such a set $S$ is also called a {\em closed set}.
The weight of the ideal cut $(S, \overline{S})$ is defined to be sum of weights of all edges 
crossing the cut i.e., $\sum_{uv: u \in S, v \not \in S} w_{uv}$.
The problem is to find an ideal cut of maximum weight.

\subsection{Overview of results and technical ideas}

We start by giving an LP formulation for the problem of finding a maximum weight ideal cut in
an edge-weighted DAG, $G$; we note that the weights can be positive as well as negative. We go on to
showing that this LP always has integral optimal solutions, hence showing that the problem is
in P (Proposition~\ref{thm:cutTime}). 
We next study the polytope obtained from the constraints of this LP (Theorem~\ref{thm.ideal}). 
We first show that the set of vertices of this polytope
is precisely the set of maximum weight ideal cuts in the DAG $G$. For this reason, we
call this the ideal cut polyhedron.
Next we characterize the edges of this polyhedron: we show that two cuts $(S, \Sbar)$ and
$(S', \Sbarc)$ are adjacent in the polyhedron iff $S \subset S'$ or $S' \subset S$.

We then study the dual of this LP. We interpret it as solving a special kind of $s$-$t$ flow 
problem in $G$ in which 
the flow on each edge has to be a least the capacity of the edge and the objective is to minimize 
the flow from $s$ to $t$. We show how to solve this flow problem combinatorially in polynomial 
time (Proposition~\ref{prop:flowAlg}).
Next, we define the notion of a residual graph for our flow problem. After finding an optimal
flow, the srongly connected components in the residual graph are shrunk to give an
unweighted DAG $D$. We show that ideal cuts in $D$ correspond to maximum weight ideal cuts in $G$ 
(Theorem~\ref{thm.max-wt}). 

We also show that the set of maximum weight ideal cuts in $G$ forms a lattice under the operations 
of set union and intersection (Theorem \ref{thm.max-wt}).

Finally, we move on to our main problem of finding a maximum weight stable matching.
We start by showing that the set $\Mcc$ of such matchings forms a sublattice $\Lcc$ of the 
lattice $\Lc$ of all stable matchings (Theorem~\ref{thm:sublattice}). 

We then give what can be regarded as the main result of our paper: a reduction from this problem 
to the problem of finding a maximum weight ideal cut in an edge-weighted DAG $G$ 
(Section \ref{sec:reduction}). This reduction goes
deep into properties of rotations and the rotation poset $\Pi$. Closed sets of $\Pi$ are
in one-to-one correspondence with the stable matchings in the lattice $\Lc$. 
In particular, if matching $M$ corresponds to closed set $S$, then starting from the boy-optimal
matching in lattice $\Lc$ we will reach matching $M$ by applying the set of rotations in $S$.

Let $R$ be the set of rotations used in $\Pi$.
We add new vertices $s$ and $t$ to $\Pi$; $s$ dominates all remaining vertices and $t$ is dominated
 by all remaining vertices. This yields the DAG $G$. The next task is to 
assign appropriate weights to the edges of $G$; this is done by using properties of
rotations. Finally, let $(S, \Sbar)$ be a maximum weight ideal cut in weighted DAG $G$, and let 
$M$ be the matching arrived at by starting from the boy-optimal
matching in lattice $\Lc$ and applying the set of rotations in $S$. Then let us say that $M$
corresponds to $S$. We show that in fact this is a one-to-one correspondence between 
maximum weight ideal cuts in $G$ and maximum weight stable matchings for the given instance
(Theorem~\ref{thm:corresponding}).

Recall the definition of (unweighted) DAG $D$ given above which was obtained from the
edge-weighted DAG $G$. As stated above, ideal cuts in the $D$ correspond to maximum weight 
ideal cuts in $G$, and hence to maximum weight stable matchings, $\Mcc$.
A vertex in $D$ corresponds to a set of vertices in $G$, and these sets form a partition of the 
set of rotations $R$. We call these sets {\em meta-rotations}. As stated above, meta-rotations
help traverse the sublattice $\Lcc$ in the same way that rotations help traverse 
the lattice $\Lc$ (Theorem \ref{thm.comb}).

\subsection{Related results}
\label{sec.related}

The works of Vande Vate \cite{Vate89} and Rothblum \cite{Rothblum92} give a linear program of
polynomial size whose vertices, also called basic feasible solutions, are precisely the set 
of stable matchings of
the given instance. Therefore, the maximum weight stable matching problem can clearly be
solved in polynomial time. As stated above, the main point of our paper is to 
obtain a combinatorial polynomial time algorithm for this problem. 

In 1987, Irving et. al. \cite{ILG87} gave a combinatorial polynomial time algorithm for the following
problem which arose in the context of obtaining an {\em egalitarian stable matching} which,
unlike the matching produced by the Gale-Shapley procedure, favors neither boys nor girls.
Each boy $b_i$ provides a preference weight $p(b_i, g_j)$ for each girl $j$ and
similarly, each girl $g_i$ provides a preference weight $p(g_i, b_j)$ for each boy $j$.
By ordering these weights, we get the preference orders for each boy and each girl. 
The problem is to find a matching that is stable under these preference orderings, say
$(b_1, g_1), (b_2, g_2), \ldots , (b_n, g_n)$, such that it maximizes (or minimizes)
$(\sum_i {p(b_i, g_i)} + \sum_i {p(g_i, b_i)})$. Clearly, this is a special case of our problem.

It is a well known that finding a minimum weight ideal cut reduces in a straightforward manner
to the minimum $s$-$t$ cut problem. However, to the best of our knowledge, a polynomial time
algorithm for maximum weight ideal cut was not known before. In particular, the reduction 
mentioned above does not give it, since the maximum $s$-$t$ cut problem is NP-hard.

Our results on the ideal cut polyhedron are motivated by results reported in \cite{GV96}.
The latter paper showed that the polyhedron obtained from the dual of the maximum $s$-$t$
flow LP does not capture all minimum $s$-$t$ cuts as vertices. They also characterized
edges of this polyhedron (it is similar to the one we obtained for ideal cuts). Then they 
went on to giving a different polyhedron whose vertices are in one-to-one correspondence with
all minimum $s$-$t$ cuts.

Our procedure of going from the edge-weighted DAG $G$ to the unweighted DAG $D$
(Theorem~\ref{thm.max-wt}) follows from the work of Picquard and Queyranne \cite{PQ80}.
Their procedure yields a DAG after performing 
max $s$-$t$ flow, computing the residual graph and shrinking strongly connected components.
As shown in \cite{PQ80}, ideal cuts in this DAG correspond to min $s$-$t$ cuts in the original
graph.

    \section{Preliminaries}

\subsection{The stable matching problem}
The stable matching problem takes as input a set of boys $B = \{b_1, b_2, \ldots , b_n\}$ and a set of girls $G = \{g_1, g_2, \ldots , g_n\}$; each person has a complete preference ranking over the set of opposite sex. 

A matching $M$ is a one-to-one correspondence between $B$ and $G$. For each pair $bg \in M$, $b$ is called the partner of $g$ in $M$ (or $M$-partner) and vice versa. 
For a matching $M$, a pair $bg \not \in M$ is said to be \emph{blocking} if $b$ is below $g$ and $g$ is below $b$, i.e., they prefer each other to their partners. A matching $M$ is \emph{stable} if there is no blocking pair in $M$.

\subsection{The lattice of stable matchings}

Let $M$ and $M'$ be two stable matchings. We say that $M$ \emph{dominates} $M'$, denoted by 
$M \preceq M'$, if every boy weakly prefers his partner in $M$ to $M'$. It is well known that 
the dominance partial order over the set of stable matchings forms a 
distributive lattice \cite{GusfieldI}, with meet and join defined as follows.
The {\em meet} of $M$ and $M'$, $M \wedge M'$, 
is defined to be the matching that results when each boy chooses his more preferred partner 
from $M$ and $M'$; it is easy to show that this matching is also stable.
The {\em join} of $M$ and $M'$, $M \vee M'$, 
is defined to be the matching that results when each boy chooses his less preferred partner 
from $M$ and $M'$; this matching is also stable. These operations distribute, i.e.,
given three stable matchings $M, M', M''$,
$$ M \vee (M' \wedge M'') = (M \wedge M') \vee (M \wedge M'') \ \ \mbox{and} \ \
M \wedge (M' \vee M'') = (M \vee M') \wedge (M \vee M'') .$$

In this paper, we will denote this lattice by $\Lc$.
It is easy to see that $\Lc$ must contain a matching, $M_0$, that dominates all others
and a matching $M_z$ that is dominated by all others.
$M_0$ is called the \emph{boy-optimal matching}, since in it, each boy is matched to his most
favorite girl among all stable matchings. This is also the {\em girl-pessimal matching}.
Similarly, $M_z$ is the {\em boy-pessimal} and \emph{girl-optimal matching}.

\subsection{Rotations help traverse the lattice}

A crucial ingredient needed to understand the structure of stable matchings is the notion of 
a rotation, which was defined by Irving \cite{Irving} and studied in detail in \cite{Irving2}. 
A rotation takes $r$ matched pairs in a fixed order, say 
$\{b_0g_0, b_1g_1,\ldots, b_{r-1}g_{r-1}\}$ and ``cyclically'' changes the mates of these $2r$ 
agents, as defined below, to arrive at another stable matching. Furthermore, it represents a minimal
set of pairings with this property, i.e, if a cyclic change is applied on any subset of these 
$r$ pairs, with any ordering, then the resulting matching has a blocking pair and is not stable.
After rotation, the boys' mates weakly worsen and the girls' mates weakly improve. Thus one can 
go from $M_0$ to $M_z$ by applying a suitable sequence of rotations (specified by the rotation poset
defined below). Indeed, this is precisely the purpose of rotations.

Let $M$ be a stable matching. For a boy $b$ let $s_M(b)$ denote the first girl $g$ on $b$'s list such that $g$ strictly prefers $b$ to her $M$-partner. Let $next_M(b)$ denote the partner in $M$ of girl $s_M(b)$. A \emph{rotation} $\rho$ \emph{exposed} in $M$ is an ordered list of pairs $\{b_0g_0, b_1g_1,\ldots, b_{r-1}g_{r-1}\}$ such that for each $i$, $0 \leq i \leq r-1$, $b_{i+1}$ is $next_M(b_i)$, where $i+1$ is taken modulo $r$. In this paper, we assume that the subscript is taken modulo $r$ whenever we mention a rotation. Notice that a rotation is cyclic and the sequence of pairs can be rotated. $M / \rho$ is defined to be a matching in which each boy not in a pair of $\rho$ stays matched to the same girl and each boy $b_i$ in $\rho$ is matched to $g_{i+1} = s_M(b_i)$. It can be proven that $M / \rho$ is also a stable matching. The transformation from $M$ to $M / \rho$ is called the \emph{elimination} of $\rho$ from $M$.

Let $\rho = \{b_0g_0, b_1g_1,\ldots, b_{r-1}g_{r-1}\}$ be a rotation. We say that $\rho$ \emph{moves $b_i$ from $g_i$ to $g_{i+1}$}, and it \emph{moves $g_i$ from $b_{i}$ to $b_{i-1}$},
for $0 \leq i \leq r-1$.


\subsection{The rotation poset}

A rotation $\rho'$ is said to precede (or dominate) another rotation $\rho$, denoted by $\rho' \prec \rho$, if $\rho'$ is eliminated in every sequence of eliminations from $M_0$ to a stable matching in which $\rho$ is exposed. Thus, the set of rotations forms a partial order via
this precedence relationship. The partial order on rotations is called \emph{rotation poset} and denoted by $\Pi$.

\begin{lemma}[\cite{GusfieldI}, Lemma 3.2.1]
	\label{lem:pre2}
	For any boy $b$ and girl $g$, there is at most one rotation that moves $b$ to $g$ or $g$ from $b$. Moreover, if $\rho_1$ moves $b$ to $g$ and $\rho_2$ moves $b$ from $g$ then $\rho_1 \prec \rho_2$.
\end{lemma}

A \emph{closed subset} is a subset of the poset such that if an element is in the subset then all of its predecessors are also included. There is a one-to-one relationship between the stable matchings and the closed subsets of $\Pi$. Given a closed subset $C$, the correponding matching $M$ is found by eliminating the rotations starting from $M_0$ according to the topological ordering of the elements in the subset. We say that $C$ \emph{generates} $M$.

\begin{lemma} [\cite{GusfieldI}, Lemma 3.3.2]
	\label{lem:computePoset}
	$\Pi$ contains at most $O(n^2)$ rotations and can be computed in polynomial time.
\end{lemma}

    \section{Maximum Weight Ideal Cuts: IP, LP and Polyhedron}
\label{sec:cuts}
In this section, we show how to find a maximum weight ideal cut using linear programming. 
We also prove some characteristics of the solution set and define a polyhedron whose vertices
are precisely the ideal cuts.

\subsection{A linear program for maximum weight ideal cut}
 

Consider the following integer program which has a variable $y_v$ for each vertex $v$ of
DAG $G = (V, E)$:
\begin{equation}
\label{ip}
\begin{aligned}
\max  &~~~ \sum_{uv \in E} w_{uv} \left( y_v - y_u \right)  \\
\ST  &~~~ y_v \geq y_u &\quad &\forall e = uv \in E \\
&~~~ y_t = 1 \\
&~~~ y_s = 0 \\
&~~~ y_v \in \{0,1\} &\quad &\forall v \in V.
\end{aligned}
\end{equation}

\begin{lemma} \label{lem:ipCorrectness}
	An optimal solution to (\ref{ip}) is a maximum weight ideal cut in $G$.
\end{lemma}
\begin{proof}
	Let 
	$ S = \{ v: y_v = 0 \}.$ 
	The set of constraints 
	\[y_v \geq y_u \quad \forall e = uv \in E\]
	guarantees that there are no edges coming into $S$. Hence, $S$ forms an ideal cut. 
	For each edge $e = uv \in E$, 
	\[
	 y_v - y_u =
	\begin{cases}
	1 &\mbox{if $u \in S$ and $v \not \in S$,} \\
	0 &\mbox{otherwise.}
	\end{cases}
	\]
	Therefore, 
	\[\sum_{e \in E} w_e \left( y_v - y_u \right) = \sum_{uv: u \in S, v \not \in S} w_{uv} = \sum_{e \text{ cross } S} w_{e}. \]
	Thus, (\ref{ip}) finds an ideal cut that maximizes the sum of weights of crossing edges as desired.
\end{proof}

Now consider the following LP relaxation of (\ref{ip}):
\begin{equation}
\label{lp}
\begin{aligned}
\max  &~~~ \sum_{uv \in E} w_{uv} \left( y_v - y_u \right)   \\
\ST   &~~~ y_v \geq y_u   &\quad& \forall e = uv \in E \\
&~~~ y_t = 1 \\
&~~~ y_s = 0. 
\end{aligned}
\end{equation}
Note that the above constraints imply $0 = y_s \leq y_v \leq  y_t = 1$ for each $v \in V$ since there is a directed path from $s$ to $v$ and a directed path from $v$ to $t$. We show how to round a solution of (\ref{lp}) to an integral solution with the same objective function value. Later on we show that any basic feasible solution of (\ref{lp}) is integral anyway.

Let $\yy$ be a (fractional) optimal solution of (\ref{lp}) and 
$\yy^*$ be an integral solution such that 
\[
y^*_v =
\begin{cases}
1 &\mbox{if $y_v > 0$,} \\
0 &\mbox{if $y_v = 0$.}
\end{cases}
\]

\begin{lemma} \label{lem:round}
	$\yy^*$ has the same objective value as $\yy$.
\end{lemma}

\begin{proof}	
Assume that $\yy$ is not integral, since otherwise the statement is trivially true.
	We will say that $y_v$ is the potential of $v$.
	Now there must exist $v \in V$ such that $0 < y_v = a < 1$. 
	Denote $S_a$ by the set of all vertices having potential $a$. 
	Let $E_{\text{in}}$ be the set of edges going into $S_a$:
	\[E_{\text{in}} = \{ uv \in E: u \not \in S_a, v \in S_a \}\]
	and $E_{\text{out}}$ be the set of edges going out of $S_a$:
	\[E_{\text{out}} = \{ uv \in E: u \in S_a, v \not \in S_a \}.\]
	
	\textbf{Claim.}
	$ \sum_{e \in E_{\text{in}}} w_e = \sum_{e \in E_{\text{out}}} w_e. $
	
	Consider adding to the potentials of all vertices in $S_a$ an amount $\delta$ where $\abs{\delta}$ is small enough so that no constraint is violated. 
	Specifically, the potential of $v \in S_a$ after modification is $y'_v = y_v + \delta$. The change in objective function along edges in $E_{in}$ is 
	\[\sum_{uv \in E_{\text{in}}} w_{uv} (y'_v - y_u)  - \sum_{uv \in E_{\text{in}}} w_{uv} (y_v - y_u) =  \sum_{uv \in E_{\text{in}}} w_{uv} (y'_v - y_v) =  \sum_{uv \in E_{\text{in}}} w_{uv} \delta. \]
	The change in objective function along edges in $E_{out}$ is
	\[\sum_{uv \in E_{\text{out}}} w_{uv} (y_v - y'_u) - \sum_{uv \in E_{\text{out}}} w_{uv} (y_v - y_u)  =  \sum_{uv \in E_{\text{out}}} w_{uv} (y_u - y'_u) = - \sum_{uv \in E_{\text{out}}} w_{uv} \delta. \]
	The total change is 
	\[ \sum_{uv \in E_{\text{in}}} w_{uv} \delta - \sum_{uv \in E_{\text{out}}} w_{uv} \delta = \delta \left(\sum_{uv \in E_{\text{in}}} w_{uv} - \sum_{uv \in E_{\text{out}}} w_{uv} \right).\]
	If $\sum_{e \in E_{\text{in}}} w_e \not = \sum_{e \in E_{\text{out}}} w_e$, we can always pick a
	 sign for $\delta$ so as to obtain a strictly better solution. Therefore, $\sum_{e \in E_{\text{in}}} w_e = \sum_{e \in E_{\text{out}}} w_e$. 
	
	
	Let $a'$ be the smallest $y$-value that is greater than $a$. $\sum_{e \in E_{\text{in}}} w_e = \sum_{e \in E_{\text{out}}} w_e$ implies that we can increase the potentials of all vertices in $S_a$ to $a'$ and obtain the same objective value. 
	The theorem follows by induction on the number of possible $y$-values. 
\end{proof}

Lemma~\ref{lem:ipCorrectness} and Lemma~\ref{lem:round} give:

\begin{proposition}
	\label{thm:cutTime}
	A maximum weight ideal cut can be found in polynomial time.
\end{proposition}

\subsection{The ideal cut polytope}

Consider the polyhedron $P$ formed by the constraints on $\yy$ in (\ref{lp}):
\begin{align*}
 y_v &\geq y_u   \quad \forall e = uv \in E \\
 y_t &= 1 \\
 y_s &= 0. 
\end{align*}

Let $n$ be the number of vertices in $G$. 
A vertex of $P$ is a feasible solution having at least $n$
linearly independent active constraints (constraints that are satisfied at equality). Let $A$ be the set of those
constraints.
Notice that in any feasible solution, $y_s = 0$ and $y_t = 1$ must be active.
Let $G_a$ be a graph such that $V(G_a) = V(G)$ and 
\[E(G_a) = \{ e: \text{ the constraint corresponding to $e$ is in $A$.}\}\]
We call $G_a$ an \emph{active graph}.

\begin{lemma} \label{lem:activeGraph}
	$G_a$ consists of two trees $T_1 \ni s$ and $T_2 \ni t$ such that $V(T_1) \cup V(T_2) = V(G)$ and $V(T_1) \cap V(T_2) = \emptyset$.
\end{lemma}

\begin{proof} We prove that $G_a$ contains no cycle and no  $s-t$ path. Since each edge of $G_a$ corresponds to a constraint in $A$, $G_a$ has at least $n-2$ edges. The lemma will then follow.
	
\textbf{Claim.} $G_a$ contains no cycle.

Assume $G_a$ contains cycle $(v_0, v_1\ldots, v_k)$. Since edges in the cycle correspond to active constraints,
$y_{v_i} = y_{v_j}$ for each edge $v_i v_j$ in the cycle. Therefore, $y_{v_0} = y_{v_1}, y_{v_1} = y_{v_2}, \ldots ,y_{v_{k-1}}= y_{v_k}$, which implies $y_{v_0} = y_{v_k}$. It follows that the set of inequalities are not independent.

\textbf{Claim.} $G_a$ contains no $s-t$ path.

Assume $G_a$ contains a path $(s, v_0\ldots, v_k, t)$.  Since edges in the path correspond to active constraints,
$y_{v_i} = y_{v_j}$ for each edge $v_i v_j$ in the path. Therefore, $y_s = y_{v_0} = \ldots = y_{v_k} = y_t$, which is a contradiction.

\end{proof}

An edge in polytope $P$ is defined by the intersection of $n-1$ linearly independent inequalities. Two vertices, also called basic feasible solutions, of the polytope are \emph{neighbors} if and only if they share an edge, i.e., the sets of inequalities that define them differ in only one inequality. Two cuts are said to be neighbors if two basic feasible solutions corresponding to them are neighbors. 

\begin{theorem}
\label{thm.ideal}
	All vertices of polyhedron $P$ are integral, and the set of vertices is precisely the 
	 set of ideal cuts. Moreover, vertices of $P$ corresponding to cuts $(S, \Sbar)$ 
	 and $(S',  \Sbarc)$ are neighbors if and only if $S \subset S'$ or $S' \subset S$. 
\end{theorem}
\begin{proof}
	By Lemma~\ref{lem:activeGraph}, $G_a$ consists of two non-intersecting trees $T_1 \ni s$ and $T_2 \ni t$.
	So $y_v = y_s = 0$ for all $v \in T_1$ and $y_u = y_t = 1$ for all $u \in T_2$. Therefore, $\yy$ is integral.
	
	Now consider an ideal cut defined by $S$. We can find a tree $T_1$ connecting all vertices in $S$, and a tree $T_2$ connecting all vertices in $V(G_a) \setminus S$. Consider the following set of inequalities:
	\begin{enumerate}
		\item $|S| - 1$ constraints corresponding to edges in tree $T_1$,
		\item $ n - |S| - 1$ constraints corresponding to edges in tree $T_2$,
		\item $y_t = 1$ and $y_s = 0$.
	\end{enumerate}
	Clearly, the set contains $n$ linearly independent inequalities, and the basic feasible solution 
	obtained by those inequalities is exactly the ideal cut $(S, \Sbar)$.

Next, we prove the second statement.
	If the cuts defined by $S$ and $S'$ are neighbors, the sets of inequalities defining them differ in only one inequality. Let $G_a$ and $G_a'$ be active graphs for $S$ and $S'$ respectively. By Lemma~\ref{lem:activeGraph}, $G_a$ consists of two trees $T_1, T_2$, 
	and $G_a'$ consists of two trees $T_1',T_2'$. 
	Moreover, $V(T_1) \cup V(T_2) = V(T_1') \cup V(T_2') = V(G)$ and $V(T_1) \cap V(T_2) = V(T_1') \cap V(T_2') = \emptyset$. 
	Since the sets of inequalities defining $S$ and $S'$ differ in only one inequality, $G_a'$ results from $G_a$ by removing an edge $e_1$ and adding an edge $e_2$. 
	Consider the graph $G'$ obtained by removing $e_1$ from $G_a$. Without loss of generality, assume that $e_1 \in E(T_1)$.
	Therefore, there exists $X \subset V(T_1)$ such that vertices in $X$ are not reachable from $s$ in $G'$. 
	By the proof of Lemma~\ref{lem:activeGraph}, $G_a'$ contains no cycle. Hence, $e_2$ can either connect $X$ to a vertex in $V(T_1) \setminus X$ or a vertex in $V(T_2)$.
	If $e_2$ connects $X$ to a vertex in $V(T_1)$, $S = V(T_1) = V(T_1') = S'$, which contradicts the fact that $S$ and $S'$ are neighbors. If $e_2$ connects $X$ to a vertex in $V(T_2)$, we have $S' = V(T_1') = V(T_1) \setminus X \subset V(T_1) = S$. 
	
	On the other direction, assume that $S' \subset S$ without loss of generality. We will give a set of active inequalities defining $S$ and a set of active inequalities defining $S'$ such that they differ in only one inequality. Let $X = S \setminus S'$. Let $T_S, T_X, T_{\overline{S'}}$ be spanning trees of $S,X$ and $V(G) \setminus S'$ respectively.
	Let $v$ be a vertex in $X$. 
	By assumption on $G$, there exists a path $Q$ from $s$ to $t$ containing $v$.
	Since $S$ and $S'$ are ideal cuts, there exist an edge $e_1 \in Q$ from $S$ to $X$ 
	and an edge $e_2 \in Q$ from $X$ to $V(G) \setminus S'$.
	Consider the set of inequalities for edges in $E(T_S) \cup  E(T_X) \cup E(T_{\overline{S'}}) \cup \{e_2\}$. These inequalities define $S$. Similarly, the inequalities for edges in $E(T_S) \cup  E(T_X) \cup E(T_{\overline{S'}}) \cup \{e_1\}$ define $S'$. The two sets of inequalities differ by only one inequality as desired. 
\end{proof}

Theorem \ref{thm.ideal} justifies calling the polytope defined in this section {\em the ideal cut 
polytope}.

\section{Maximum Weight Ideal Cuts: Combinatorial Algorithm}
\label{sec:lattice}

\subsection{The set of maximum weight ideal cuts forms a lattice}

We first prove the following fact.

\begin{lemma}
\label{lem.max-wt}
		If $S$ and $S'$ are two subsets defining 
	maximum weight ideal cuts in $G$ then $S \cup S'$ and $S \cap S'$ also
	define maximum weight ideal cuts.
\end{lemma}

\begin{proof}
	Let $E_1$ be the set of edges going from $S \cap S'$ to $S \setminus S'$:
	\[E_1 = \{uv \in E: u \in S \cap S', v \in S \setminus S' \}.\]
	Let $E_2$ be the set of edges going from $S \cap S'$ to $S' \setminus S$:
	\[E_2 = \{uv \in E: u \in S \cap S', v \in S' \setminus S \}.\]
	Let $E_3$ be the set of edges going from $S \setminus S'$ to $V \setminus (S' \cup S)$:
	\[E_3 = \{uv \in E: u \in S \setminus S', v \in V \setminus (S' \cup S) \}.\]
	Let $E_4$ be the set of edges going from $S' \setminus S$ to $V \setminus (S' \cup S)$:
	\[E_4 = \{uv \in E: u \in S' \setminus S, v \in V \setminus (S' \cup S) \}.\]
	Let $E_5$ be the set of edges going from $S' \cap S$ to $V \setminus (S' \cup S)$:
	\[E_5 = \{uv \in E: u \in S' \cap S, v \in V \setminus (S' \cup S) \}.\]
	Note that there are no edges going between $S' \setminus S$ and $S \setminus S'$.
	
	Therefore, the weight of the ideal cut defined by $S$ is
	\[ w(S) =  \sum_{e \in E_2} w_e  + \sum_{e \in E_3} w_e + \sum_{e \in E_5} w_e.  \]
	The size of the ideal cut defined by $S'$ is
	\[ w(S') = \sum_{e \in E_1} w_e  + \sum_{e \in E_4} w_e + \sum_{e \in E_5} w_e.  \] 
	Since both cuts are maximum weight ideal cuts, 
	\[ \sum_{e \in E_2} w_e  + \sum_{e \in E_3} w_e = \sum_{e \in E_1} w_e  + \sum_{e \in E_4} w_e.\]
	
	We will show that $\sum_{e \in E_1} w_e = \sum_{e \in E_3} w_e$ and $\sum_{e \in E_2} w_e = \sum_{e \in E_4} w_e$. 
	Assume that  $\sum_{e \in E_1} w_e < \sum_{e \in E_3} w_e$ and $\sum_{e \in E_2} w_e < \sum_{e \in E_4} w_e$. Then 
	the weight of the cut defined by $S \cup S'$ is 
	\[w(S \cup S') = \sum_{e \in E_3} w_e  + \sum_{e \in E_4} w_e + \sum_{e \in E_5} w_e > w(S). \]
	Similarly, if $\sum_{e \in E_1} w_e > \sum_{e \in E_3} w_e$ and $\sum_{e \in E_2} w_e > \sum_{e \in E_4} w_e$, 
	\[w(S \cap S') = \sum_{e \in E_1} w_e  + \sum_{e \in E_2} w_e + \sum_{e \in E_5} w_e > w(S). \]
	Therefore,  $\sum_{e \in E_1} w_e = \sum_{e \in E_3} w_e$, $\sum_{e \in E_2} w_e = \sum_{e \in E_4} w_e$ and 
	\[w(S \cup S') = w(S \cap S') = w(S) = w(S').\]
\end{proof}

Lemma \ref{lem.max-wt} gives:

\begin{theorem}
\label{thm.max-wt}
		The set of maximum weight ideal cuts forms a lattice under 
the operations of union and intersection.
\end{theorem}

\subsection{A flow problem in which capacities are lower bounds on edge-flows}

To unveil the underlying combinatorial structure, we consider the dual program of (\ref{lp}). First, (\ref{lp}) can be rewritten as: 

\begin{equation}
\label{another-lp}
\begin{aligned}
\max  &~~~ \sum_{e \in E} w_e z_e  \\
\ST   &~~~ z_e = y_v - y_u &\quad& \forall e = uv \in E \\
&~~~ y_t - y_s = 1 \\
&~~~ z_e \geq 0 &\quad& \forall e \in E.
\end{aligned}
\end{equation}

Let $f_{uv}$ be the dual variable corresponding to edge $uv$. The dual linear program is: 
\begin{equation}
\label{flow-lp}
\begin{aligned}
\min  &~~~ f_{ts}  \\
\ST   &~~~ \sum_{u:uv \in E} f_{uv} =  \sum_{u:vu \in E} f_{vu}  &\quad& \forall v \in V \\
&~~~ f_{uv} \geq w_{uv}  &\quad& \forall uv \in E
\end{aligned}
\end{equation}

We show that (\ref{flow-lp}) can be interpreted as a flow problem. 
To be precise, $f_{uv}$ represents the flow value on edge $uv$. 
The first set of inequalities guarantees flow conservation at each vertex.  
The second set of inequalities says that there is a lower bound $w_{uv}$ on the amount of flow on $uv$. 
Note that $f_{uv}$ as well as $w_{uv}$ can be negative.
 
The problem is to find a minimum circulation in the graph obtained by adding an infinite capacity edge from $t$ to $s$ to $G$. 
Equivalently, without the introduction of $ts$, the problem can be seen as finding a minimum 
flow from $s$ to $t$ in $G$.


We give a combinatorial algorithm to solve the above flow problem. 
The high level idea is to first find a feasible flow, i.e, a flow $\ff$ satisfying all inequalities. We then push flow back as much as possible from $t$ to $s$, while maintaining flow feasibility.

It is easy to see that the routine in Figure~\ref{alg:feasible} gives us a feasible flow.
At the end of the routine, the value of flow from $s$ to $t$ is at most $nW$ where $W = \max_{e} |w_e|$. 
\begin{figure}
	\begin{algorithmic}
			\While{there exists an edge $uv \in E$ such that $w_{uv} >0$ and $f_{uv} < w_{uv}$}
			\State Find a path $Q$ from $s$ to $t$ containing $uv$.
			\State Send flow of value $w_{uv}$ along $Q$.
			\EndWhile
	\end{algorithmic}
	\caption{Routine for Finding a Feasible Flow.}
	\label{alg:feasible} 
\end{figure}

To push flow back from $t$ to $s$, we construct the following residual graph $G_{\ff}$ for a feasible flow $\ff$.
Since $\ff$ is feasible, $f_{uv} \geq w_{uv} \, \forall uv \in E$. 
For each $uv \in E$ such that $f_{uv} > w_{uv}$, we create a residual edge from $v$ to $u$ with capacity $f_{uv} - w_{uv} > 0$. Notice that the capacity on $vu$ is exactly the amount we can push back on $uv$ without violating the lowerbound constraint. Finally, all edges in $E$ still have infinite capacity. 

Let $\xx > 0$ be a feasible flow in $G_{\ff}$. In other words, $\xx$ satisfies flow conservation and capacity constraints. Let $\ffbar = \ff \oplus \xx$ be a flow constructed as follows: 
\begin{equation}
\label{eq:augment}
\fbar_{uv} =
\begin{cases}
f_{uv} + x_{uv} - x_{vu} &\mbox{if $vu$ is an edge in $G_{\ff}$,} \\
f_{uv} + x_{uv} &\mbox{otherwise.}
\end{cases}
\end{equation}

\begin{lemma}
	\label{lem:flowFeasibility}
	$\ffbar$ is a feasible solution to (\ref{flow-lp}).	
\end{lemma}
\begin{proof}
	Flow conservation is satisfied trivially. It suffices to show that no lower bound constraint is violated. Consider 2 cases: 
	\begin{itemize}
		\item if $vu$ is an edge in $G_{\ff}$, the capacity of $vu$ is $f_{uv} - w_{uv}$. Therefore, $f_{uv} + x_{uv} - x_{vu} \geq f_{uv} - (f_{uv} - w_{uv}) = w_{uv}$.
		\item if $vu$ is not an edge in $G_{\ff}$, $f_{uv} + x_{uv} \geq f_{uv} = w_{uv}$.
	\end{itemize}
\end{proof}

\begin{lemma}
	\label{lem:optimality}
	$\ff$ is an optimal solution of (\ref{flow-lp}) if and only if there is no path from $t$ to $s$ in $G_{\ff}$.
\end{lemma}

\begin{proof}
	Suppose that there exists a path from $t$ to $s$ in $G_{\ff}$. Sending flow 
	along the path gives a feasible flow by Lemma~\ref{lem:flowFeasibility}. 
	Moreover, the objective function has a smaller value. Therefore, $\ff$ is not an optimal solution of (\ref{flow-lp}).
	
	If there is no path from $t$ to $s$, let $T$ be the set of vertices that are reachable from $t$ by a path in $G_{\ff}$: 
	\[ T = \{v \in V: \exists \text{ path } p \text{ from } t \text{ to } v \text{ in } G_{\ff} \}.\]
	Consider $uv \in E$ such that $v \in T$ and $u \not \in T$. Since $vu$ is not an edge in $G_{\ff}$ and $\ff$ is feasible, $f_{uv} = w_{uv}$. 
	
	Let $\yy$ be the primal solution such that $y_v = 1$ for all $v \in T$ and $y_v = 0$ otherwise. With respect to $\yy$, $z_{uv} > 0$ if and only if $v \in T$ and $u \not \in T$ if and only if ${f}_{uv} = w_{uv}$. Therefore, $\ff$ and $\yy$ satisfy complementarity. Hence, $\ff$ is an optimal solution of (\ref{flow-lp}).
\end{proof}

By Lemma~\ref{lem:flowFeasibility} and Lemma~\ref{lem:optimality}, 
a natural algorithm, given a feasible flow $\ff$, is the following: Iteratively find a path from $t$ to $s$ in $G_{\ff}$. If there exists such a path, send maximal flow back on this path
without violating feasibility, 
 update $\ff$ and repeat. Otherwise, $\ff$ is an optimal flow by Lemma~\ref{lem:optimality}.

Notice that the above routine is very similar to the Ford–Fulkerson algorithm for finding maximum $s$-$t$ flow. A more straight forward way is to compute a maximum flow in $G_{\ff}$ for a feasible flow $\ff$ as shown in Figure~\ref{alg:flow}.

\begin{figure}
		\begin{algorithmic}
			\State Find a feasible flow $\ff$.
			\State Find a maximum flow $\xx$ from $t$ to $s$ in $G_{\ff}$.
			\State Return $\ffbar = \ff \oplus \xx$ as shown in \ref{eq:augment}.
		\end{algorithmic}
	\caption{Combinatorial Algorithm for Finding Flow.}
	\label{alg:flow} 
\end{figure}

\begin{proposition}
\label{prop:flowAlg}
	The algorithm in Figure~\ref{alg:flow} finds an optimal flow for (\ref{flow-lp}).
\end{proposition}

\begin{proof}
By Lemma~\ref{lem:optimality}, it suffices to show that	there is no path from $t$ to $s$ in $G_{\ffbar}$ if and only if $\xx$ is a maximum flow from $t$ to $s$ in $G_{\ff}$. 

If there exists a path from $t$ to $s$ in $G_{\ffbar}$, then there exists $\xx'$ such that $\ffbar \oplus \xx' = (\ff \oplus \xx) \oplus \xx' = \ff \oplus (\xx + \xx')$ is a flow of from $s$ to $t$ smaller value than $\ffbar$. Therefore, $\xx + \xx'$ is a flow from $t$ to $s$ in $G_{\ff}$ of greater value than $\xx$, which is a contradiction. 

If $\xx$ is not a maximum flow, there exists $\xx'$ such that $\xx + \xx'$ is a feasible flow from $t$ to $s$ in $G_{\ff}$ of greater value. Therefore, there exists a path from $t$ to $s$ in $G_{\ff \oplus \xx}$.
\end{proof}

\subsection{Generating all maximum weight ideal cuts}
\label{subsec:structure}

The process is similar to finding the Picard-Queyranne structure, whose ideal cuts
are in one-to-one correspondence with the minimum $s$-$t$ cuts in a graph. 
Given an optimal flow solution $\ff$, we shrink the strongly connected components of $G_{\ff}$.
The resulting graph is a DAG $D$. Now,
ideal cuts in $D$ are in one-to-one correspondence with maximum weight cuts in the original graph.
Hence we get:

\begin{theorem} \label{thm:idealcutdag}
	There is a combinatorial polynomial time algorithm for constructing a
	DAG $D$ such that an ideal cut in $D$ bijectively corresponds to a maximum weight ideal cut in $G$. 
\end{theorem}

    \section{Maximum Weight Stable Matching Problem}

\subsection{The reduction}
\label{sec:reduction}

Given an instance $I$ of maximum weight stable matching problem, we show how to obtain an instance 
$J$ of maximum weight ideal cut problem such that there is a bijection between the set of 
solutions to $I$ and those to $J$. 

For this purpose, we show how to construct a DAG $G$ with an edge-weight function $w$.
We start with the rotation poset $\Pi$ that generates all stable matchings for $I$. This can be 
obtained in polynomial time by Lemma~\ref{lem:computePoset}. Next,  
we construct an edge-weighted DAG $G$ as follows: 
\begin{enumerate}
	\item Keep all vertices and edges in the natural DAG representation of $\Pi$. Let $v_i$ be the vertex that corresponds to $\rho_i$. 
	\item Add a source $s$ and an edge from $s$ to every $v_i$ such that $\rho_i$ is not dominated by any other rotation.
	\item Add a sink $t$ and an edge from every $v_i$ to $t$, such that $\rho_i$ does not dominate any other rotation.
\end{enumerate} 

Next, consider all pairs $bg$ that appears in the stable matchings of the given instance.
Ignore a pair $bg$ if it appears in all stable matchings. With each of the remaining pairs $bg$, 
we associate a directed path $P_{bg}$ in $G$ as follows:

\begin{itemize}
	\item 
	{\bf Case 1,} $bg \in M_0, bg \not \in M_z$: There exists a rotation $\rho_i$ that moves $b$ away from $g$. Choose $P_{bg}$ to be an arbitrary path in $G$ from $s$ to $v_i$.
	\item 
	{\bf Case 2,} $bg \in M_z, bg \not \in M_0$: There exists a rotation $\rho_i$ that moves $b$ to $g$. Choose $P_{bg}$ to be an arbitrary path in $G$ from $v_i$ to $t$.
	\item 
	{\bf Case 3,} $bg \not \in M_0, bg \not \in M_z$: There exist a rotation $\rho_i$ moving $b$ to $g$ and a rotation $\rho_j$ moving $b$ from $g$. By Lemma~\ref{lem:pre2}, $\rho_i$ dominates $\rho_j$, and therefore there is at least one path in $G$ from $v_i$ to $v_j$. Choose $P_{bg}$ to be an arbitrary such path.
\end{itemize}

Finally, we assign weights to the edges of $G$ as follows. Initialize all edge weights to 0.
Then, for each pair $bg$, we add $w_{bg}$ to the weights of all edges in $P_{bg}$. We also say 
that $w_{P_{bg}} = w_{bg}$

Clearly, an ideal cut in $G$ corresponds to a closed subset in $\Pi$. To be precise, for a non-empty vertex set $S$ such that 
$s \in S$ and there are no incoming edges to $S$, 
\[ C = \{ \rho_i: v_i \in S \setminus \{s\} \}\] 
is clearly a closed subset in $\Pi$.
We prove a simple yet crucial lemma:
\begin{lemma} \label{lem:corresponding}
	$S$ cuts $P_{bg}$ if and only if the matching generated by $C$ contains $bg$. 
\end{lemma}
\begin{proof}
	We will use the following key observation:  for any pair $u,v$ of vertices in a DAG such that there exist paths from $u$ to $v$, an ideal cut separates $u$ and $v$ if and only if
 	it cuts each of these paths exactly one. We consider 3 cases:
	\begin{itemize}

		\item 
		{\bf Case 1,} $bg \in M_0, bg \not \in M_z$: There exists a unique rotation $\rho_i$ that moves $b$ away from $g$. $S$ cuts $P_{bg}$ if and only if $C$ does not contain $\rho_i$.
This happens if and only if the matching generated by $C$ contains $bg$.
		\item 
		{\bf Case 2,} $bg \not \in M_0, bg \in M_z $: There exists a unique rotation $\rho_i$ that moves $b$ to $g$. $S$ cuts $P_{bg}$ if and only if $C$ contains $\rho_i$.
This happens if and only if the matching generated by $C$ contains $bg$.
		\item 
		{\bf Case 3,} $bg \not \in M_0, bg \not \in M_z$: There exist a unique rotation $\rho_i$ moving $b$ to $g$ and a unique rotation $\rho_j$ moving $b$ from $g$. $S$ cuts $P_{bg}$ if and only if $C$ contains $\rho_i$ and does not contain $\rho_j$.
This happens if and only if the matching generated by $C$ contains $bg$.
	\end{itemize}
\end{proof}

\begin{theorem}
\label{thm:corresponding}
The maximum weight stable matchings in $I$ are in one-to-one correspondence with the maximum weight ideal cuts in $J$.
\end{theorem}

\begin{proof}
	We show that the weight of an ideal cut generated my $S$ is equal to the weight of the 
	matching generated by $C$. By Lemma~\ref{lem:corresponding}, 
	\begin{align*}
		w(S) &= \sum_{e = uv: u \in S, v \not \in S} w_e 
		     = \sum_{e = uv: u \in S, v \not \in S} ~ \sum_{e \in P_{bg}} w_{bg} \\
		     &= \sum_{S \text{ cuts } P_{bg}} w_{bg} 
		     = \sum_{bg \in \text{ the matching generated by } C}  w_{bg}.
	\end{align*}
	The theorem follows.
\end{proof} 

\subsection{The sublattice, and using meta-rotations to traversing it}

 By Theorem \ref{thm.max-wt} and Theorem~\ref{thm:corresponding} we get:

\begin{lemma}
	If $M$ and $M'$ are maximum stable matchings in $\Mc$ then so are $M \vee M'$ and $M \wedge M'$.
\end{lemma}

This gives:

\begin{theorem}
	\label{thm:sublattice}
	The set of maximum weight stable matchings forms a sublattice $\Lcc$ of the lattice $\Lc$.
\end{theorem}

We next give the notion of a meta-rotation. These help traverse the sublattice $\Lcc$ in the 
same way that rotations help traverse the lattice $\Lc$.
Let $R$ be the set of all rotations used in the rotation poset $\Pi$. 
Let $G$ be the graph obtained from $\Pi$ by adding vertices $s$ and $t$ and 
assigning weights to edges, as described in Section \ref{sec:reduction}. 
Let $D$ be the DAG constructed in
Theorem~\ref{thm:idealcutdag}; ideal cuts in $D$ correspond to a maximum weight ideal cuts in $G$.
A vertex, $v$, in $D$ corresponds to a set of vertices in $G$. Hence we may view $v$ as 
a subset of the rotations in $R$; clearly, the subsets represented by the set of all vertices in $D$ 
form a partition of $R$.

By analogy with the rotation poset $\Pi$, let us represent $D$ by $\Pibar$ and call it
the {\em meta-rotation poset}.
Each vertex in $\Pibar$ (and $D$) is a subset of $R$ and is called a {\em meta-rotation}.
Let $S$ be the element in $\Pibar$ containing $s$, and $T$ be the element in $\Pibar$ 
containing $t$.
For any closed subset, $P$, of $\Pibar$, let $R_P$ be the set of all rotations 
contained in the meta-rotations of $P$.
Eliminating these rotations starting from $M_0$, according to the topological ordering of the rotations given in $\Pi$, we arrive at a maximum weight stable matching, say $M_P$.
In this manner, the meta-rotations help us traverse the sublattice. Combining with
Proposition~\ref{prop:flowAlg} and Theorem~\ref{thm:corresponding}, we get:

\begin{theorem}
\label{thm.comb}
There is a combinatorial polynomial time algorithm for finding a
maximum weight stable matching. The meta-rotation poset
	$\Pibar$ can also be constructed in polynomial time. Each closed subset of $\Pibar$ containing $S$ and not containing $T$ generates a maximum weight stable matching.
\end{theorem}

The running time of the algorithm described above is dominated by the time required to find a
max-flow in the graph obtained from $\Pi$ which has $O(n^2)$ vertices.

\subsection{Further applications of the structure}

\subsubsection{Finding boy-optimal and girl-optimal matchings in $\Lcc$}

Notice that for two closed subsets $C$ and $C'$, the matching generated by $C$ dominates the matching generated by $C'$ if $C \subset C'$. Hence we have:

\begin{lemma}
	The closed subset containing only the meta-rotation $S$ generates the boy-optimal stable 
	matching and the one
	containing all meta-rotations other than $T$ generates the girl-optimal stable 
	matching in the sublattice $\Lcc$.
\end{lemma}

\subsubsection{Bi-objective stable matching}

In the {\em bi-objective stable matching problem} we are given sets $B$ and $G$, of $n$ boys and
$n$ girls and, for each boy and each girl,
a complete preference ordering over all agents of the opposite sex. 
However, unlike the maximum weight stable matching problem,
we are given two weight functions
$\ww^{(1)},\ww^{(2)}: \ B \times G \rightarrow \mathbb{R}$. The problem is to find a stable matching $M$ that 
maximizes 
$ \sum_{bg \in M} w^{(2)}_{bg} $
among the ones maximizing  
$\sum_{bg \in M} w^{(1)}_{bg}.$

To solve this problem, first we find a poset $\Pibar$ that generates the set of stable matchings maximizing $\sum_{bg \in M} w^{(1)}_{bg}$.
Then we form a maximum ideal cut instance in the same way as described in section~\ref{sec:reduction} with respect to $\ww^{(2)}$. 
Let $G$ be the DAG in the instance. 
For each meta-rotation $V$ in $\Pibar$, contract all vertices 
in $G$
corresponding to the rotations in $V$.
Let $\overline{G}$ be the resulting graph.
By a similar argument to the one in Section~\ref{sec:reduction}, we have: 
\begin{lemma}
The maximum weight ideal cuts in $\overline{G}$ are in one-to-one correspondence with the solutions of bi-objective stable matching problem.
\end{lemma}

	\bibliographystyle{alpha}
	\bibliography{refs}
\end{document}